\documentclass[11pt]{article}
\usepackage{a4wide,amsmath,amssymb,amsthm}

\usepackage{enumerate}
\usepackage{latexsym,epsfig,graphicx}
\usepackage{subfigure}
\usepackage{booktabs}
\usepackage{multirow}
\usepackage{longtable}

\usepackage{algorithm}
\usepackage{algorithmic}

\newtheorem{Obs}{Observation}
\newtheorem{Prop}{Proposition}

\title{An effective decomposition approach and heuristics to generate spanning trees with a small number of branch vertices}

\author{Rafael A. Melo{\thanks{Departamento de Ci\^{e}ncia da Computa\c{c}\~{a}o, Instituto de Matem\'{a}tica, Universidade Federal da Bahia. Av. Adhemar de Barros, s/n, Salvador, BA 40170-110, Brazil.  ({\tt melo@dcc.ufba.br}). Work of this author was supported by the Brazilian National
Council for Scientific
and Technological Development (CNPq).}} \and Phillippe Samer \thanks{Departamento de Ci\^{e}ncia da Computa\c{c}\~{a}o, Universidade Federal de Minas Gerais, Av. Ant\^onio Carlos 6627 - Belo Horizonte, MG 31270-010, Brazil. ({\tt samer@dcc.ufmg.br})} \and Sebasti\'{a}n A. Urrutia \thanks{Departamento de Ci\^{e}ncia da Computa\c{c}\~{a}o, Universidade Federal de Minas Gerais, Av. Ant\^onio Carlos 6627 - Belo Horizonte, MG 31270-010, Brazil. ({\tt surrutia@dcc.ufmg.br})}}

\usepackage{epsfig,lscape}
\usepackage{subfigure}

\LTcapwidth=\textwidth

\begin{document}

\maketitle

\begin{abstract}
\noindent Given a graph $G=(V,E)$, the \textit{minimum branch vertices problem} consists in finding a spanning tree $T=(V,E')$ of $G$ minimizing the number of vertices with degree greater than two.
We consider a simple combinatorial lower bound for the problem, from which we 
propose a decomposition approach. The motivation is to break down the problem into several smaller subproblems which are more tractable computationally, and then recombine the obtained solutions to generate a solution to the original problem.
We also propose effective constructive heuristics to the problem which take into consideration the problem's structure in order to obtain good feasible solutions.
Computational results show that our decomposition approach is very fast and can drastically reduce the size of the subproblems to be solved. This allows a branch and cut algorithm to perform much better than when used over the full original problem. 
The results also show that the proposed constructive heuristics are highly efficient and generate very good quality solutions, outperforming other heuristics available in the literature in several situations.
\\


\end{abstract}

\section{Introduction}

Spanning tree problems with constraints and/or objective functions concerning 
the degrees of the vertices in the tree arise very 
often in the context of communication networks. In this paper we deal with the 
problem of finding spanning trees which minimize the number of branch vertices 
(vertices with degree greater than two), known in the literature as the 
\textit{minimum branch vertices problem} (MBV).

The MBV was introduced in Gargano \emph{et al.}~\cite{Garetal04,Garetal02}, motivated by the context of optical networks using a technology which allows a specific type of switch to replicate a signal by splitting light. 
Such a switch is required if a connection arrives at a given vertex of the network and has to be multicasted to different vertices. Therefore, a branch vertex in a network tree would imply the use of such a sophisticated switch.
However, the availability of these switches is usually limited due to their costs, and one is asked to minimize their need in a network topology. 

Several authors studied different theoretical properties concerning the number of branch vertices in a graph. 
Gargano \emph{et al.}~\cite{Garetal04,Garetal02} showed the problem to be NP-Hard, 
besides some additional complexity results. 
They also considered bounds on the number of branch vertices and certain algorithmic 
and combinatorial aspects regarding the existence of structures such as spanning spiders (trees with at most one branch vertex).
Matsuda \emph{et al.}~\cite{Matetal14} give bounds on the number of branch vertices in claw-free graphs.

Chimani and Spoerhase~\cite{ChiSpo13} developed a 6/11-approximation algorithm for the complement of the MBV, namely the maximum path-node spanning tree problem, which aims at finding a spanning tree maximizing the number of vertices with degree at most two.

Different integer programming and heuristic approaches were proposed to the MBV. 
Carrabs \emph{et al.}~\cite{Caretal13} studied integer programming formulations and a Lagrangian relaxation approach able to generate good feasible solutions.
    Cerrone \emph{et al.}~\cite{Ceretal14} showed the relation between three problems: the minimum branch vertices, the minimum degree sum problem and the minimum leaves problem. They also presented an evolutionary algorithm which could be applied to these problems.
Cerulli \emph{et al.}~\cite{Ceretal09} proposed three heuristics to the problem, exploiting vertex weighting, vertex coloring, and an approach combining the previous two.
Silva \emph{et al.}~\cite{Siletal14}  proposed an iterative refinement local search procedure. Almeida \emph{et al.}~\cite{Almetal14} used a constructive heuristic based on breadth-first search, together with some changes in the procedure described in~\cite{Siletal14}.
\"Oncan~\cite{Onc14} proposed inequalities to improve the linear relaxation bounds obtained with a polynomial size integer programming formulation to the problem. Rossi \emph{et al.}~\cite{Rosetal14} implemented a hybrid approach, combining a branch and cut algorithm with a tabu search heuristic, which could be applied to the MBV and to the minimum degree sum problem. Computational experiments showed that their approach outperformed the other integer programming techniques available in the literature.
Very recently, Mar\'in~\cite{Mar15} proposed a preprocessing technique, valid inequalities and a heuristic algorithm, which combined could solve to optimality several instances used in the literature.

Our work aims at the development of an effective approach to solve the MBV.
The remainder of the paper is organized as follows. 
Section \ref{sec:problem} gives a formal statement of the problem together with an integer programming formulation containing an exponential number of inequalities. 
In Section \ref{sec:lowerbound} we give a simple algorithm to derive a combinatorial lower bound, which also determines important information that will be used in our graph decomposition approach.
In Section \ref{sec:graphdecomposition} we present the decomposition method, which aims at dividing the problem into several smaller subproblems 
that are hopefully more manageable computationally.
Two constructive heuristics, which try to take advantage of the problem's structure and of what is expected of a good solution, are described in Section \ref{sec:heuristics}. In Section \ref{sec:computation}, we present computational results considering the proposed approaches together with a branch and cut algorithm using the formulation described in Section \ref{sec:problem}.
Some final remarks are discussed in Section \ref{sec:finalremarks}.
\section{Problem definition and formulation}
\label{sec:problem}

Let $G=(V,E)$ be a simple, connected, undirected graph with a set $V$ of vertices and a set $E$ of edges.
Denote $d_G(v)$ the degree of $v \in V$ in $G$. 
A vertex $v \in V$ is a branch vertex if it has degree greater than two, i.e. $d_G(v)>2$.
Given the graph $G$, the \textit{minimum branch vertices problem} (MBV) consists in finding a spanning tree $T=(V,E')$, with $E' \subseteq E$, minimizing the number of branch vertices.

We now present an integer programming formulation for the problem containing an exponential number of constraints. Consider the variables:
\[ x_{uv} = \left\{
   \begin{array}{l}
      1, \textrm{ if edge $(u,v) \in E$ is part of the tree,}\\      
      0, \textrm{ otherwise;}
   \end{array}
   \right. 
\]
\[ y_{v} = \left\{
   \begin{array}{l}
      1, \textrm{ if vertex $v \in V$ is a branch vertex,}\\      
      0, \textrm{ otherwise.}
   \end{array}
   \right. 
\]
Using the variables just described, the MBV can be formulated as:
\begin{alignat}{2}
z = \min&\quad  \sum_{v \in V} y_v \label{obj}\\
\text{s.t.}&\quad \nonumber \\
&  \sum_{(u,v)\in E} x_{uv} = |V|-1, \quad \label{const-01} \\
&  \sum_{(u,v)\in E: u,v \in S} x_{uv} \leq |S|-1  \ \ \  \quad && {\rm for \ } S \subset V,  \label{const-02}\\
&  \sum_{u:(u,v)\in E} x_{uv} \leq 2 + (d_G(v) -2)y_v \ \ \ \quad && {\rm for \ } v \in V,  \label{const-03}\\
& x \in \{0,1 \}^{|E|}, \quad \ y \in \{0,1\}^{|V|}. \label{const-04}
\end{alignat}
The objective function (\ref{obj}) minimizes the number of branch vertices.
Constraints (\ref{const-01}) specify that exactly $|V|-1$ edges are selected.
Constraints (\ref{const-02}) are subtour elimination constraints; note that we only need to consider subsets $S$ such that $|S| > 2$.
Constraints (\ref{const-03}) set the branch vertex variable associated to vertex $v$ to one in case more than two edges are incident to $v$.
Constraints (\ref{const-04}) are integrality requirements on the variables.
It is worth mentioning that we can fix the $y_v$ variable at zero for any vertex $v \in V$ with degree lower or equal to two.
\section{A combinatorial lower bound}
\label{sec:lowerbound}

In this section, we present an algorithm to calculate a combinatorial lower bound to the MBV problem, which also determines some important data to be used in a graph decomposition approach.
Let $s(G)$ be the minimum number of branch vertices in any spanning tree of $G$.
An articulation point (or cut vertex) is a vertex whose removal from the graph increases the number of components in the graph.
A straightforward lower bound $\underline{s}(G)$ on the value of $s(G)$ can be calculated based on a simple observation of Gargano et al. \cite{Garetal02}: any articulation point whose removal induces at least three connected components is necessarily a branch vertex in any spanning tree of $G$. We denote such vertices \textit{obligatory branches} and define $L_o$ to be the set of obligatory branches of $G$.

Let $G-v$ be the subgraph of $G$ obtained by removing the vertex $v$ together with all its incident edges. We denote by $\alpha(v)$ the number of connected components of $G-v$. Note that $\alpha(v) \geq 2$ for an articulation point $v$, while $\alpha(v) \geq 3$ for any obligatory branch.
Algorithm~\ref{alg:lowerbound} calculates the lower bound $\underline{s}(G)$ on the number of branch vertices, the set of obligatory branches $L_o$ and the value of $\alpha(v)$ for every obligatory branch $v \in L_o$.

\begin{algorithm}[!ht]
\caption {Obligatory branches lower bound}
\label{alg:lowerbound}
\begin{algorithmic}[1]
	\STATE $L_o \leftarrow \emptyset$
	\STATE $V' \leftarrow$ list of articulation points in $G$
	\FORALL {$v \in V'$}
		\STATE $c \leftarrow$ number of connected components of $G-v$
		\IF {$c\geq 3$}
			\STATE $L_o \leftarrow L_o + \{ v\}$
			\STATE $\alpha(v) \leftarrow c$
		\ENDIF
	\ENDFOR
	\STATE $\underline{s}(G) \leftarrow |L_o|$
\end{algorithmic}
\end{algorithm}

\begin{Obs}
Algorithm \ref{alg:lowerbound} runs in $O(|V|+|E|)$.
\end{Obs}
\begin{proof}
The list of articulation points $V'$ can be obtained in step 2 using a depth-first search (DFS) algorithm augmented with the opening time of each vertex, \emph{i.e.} its visiting index in the DFS tree. 
In particular, we identify an \emph{articulation} at a given vertex $u$ if, when the search at a neighbor $v$ returns, the oldest index (opened the earliest) reachable from $v$ is no smaller than the opening time of $u$. Finding $i$ articulations at a vertex means that the search started from it $i$ times, and that $i+1$ components are left after removing it, which is precisely the information used later in step 4. 

The DFS algorithm runs in $O(|V|+|E|)$ when an adjacency list is used.
The \textbf{for} loop (lines 3-9) is executed $|V'|$ times, and all operations within it are performed in constant order of complexity. Therefore the algorithm runs in $O(|V|+|E|) + O(|V'|)$ $=$ $ O(|V|+|E|)$.
\end{proof}

Given that the MBV problem is NP-Hard, one should note that the lower bound obtained using Algorithm \ref{alg:lowerbound} is not necessarily tight.
It is actually possible to find situations in which this occurs even for very small graphs.
For instance, consider the graph illustrated in Figure \ref{fig:exampleboundnottight}, which does not have any obligatory branch vertex. Nevertheless, it is not difficult to check that either $b$ or $d$ should be a branch vertex in a spanning tree of the considered graph.
\begin{figure}[!ht]
\centering
	\epsfig{file=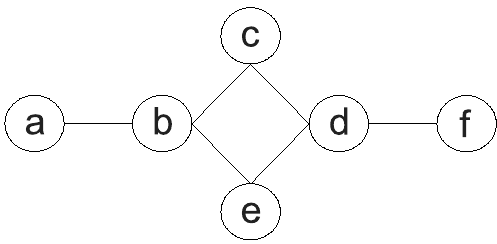,scale = 1.0}
\caption{Example in which $\underline{s}(G)=0<s(G)=1$}
\label{fig:exampleboundnottight}
\end{figure}

\section{Graph decomposition approach}
\label{sec:graphdecomposition}

In this section, we present a graph decomposition approach to the minimum branch vertices problem based on two different properties of a graph. 
The first one is grounded on the existence of obligatory branch vertices.
The second builds on the existence of cut edges (or bridges): an edge whose removal increases the number of connected components of a graph.
This decomposition approach aims at dividing the original MBV problem into smaller subproblems, that can be solved separately and then recombined in order to generate a solution to the original problem. We remark that Mar\'in~\cite{Mar15} proposed a preprocessing routine, which included the identification of cut edges, but that information was not used in a decomposition method by the author.

\subsection{Obligatory branches based decomposition}

This first decomposition is based on the fact that an obligatory branch $v\in V(G)$ is going to be a branch vertex in any spanning tree of a graph $G$.
Starting from $G=(V,E)$, we build a new graph $G_o=(V_o,E_o)$ as follows.
Initially, $V_o=V$ and $E_o=E$, such that $G_o=(V,E)$. Next, for each obligatory branch $v \in L_o$, create $\alpha(v)$ new vertices $\{v_1,\ldots,v_{\alpha(v)}\}$ and add them to $V_o$. Let $N_i(v)$ be the set of vertices to which $v$ is adjacent in the $i-th$ connected component of the subgraph $G-v$. For each $ i \in \{1 ,\ldots,\alpha(v) \}$, create in $E_o$ edges from $v_i$ to every vertex in $N_i(v)$. After that, remove $v$ from $V_o$ and all its adjacent edges from $E_o$.

The obligatory branches based decomposition is illustrated in Figure \ref{fig:decompositionoblbr} for a graph with two obligatory branches.

\begin{figure}[!ht]
\centering
\subfigure[Original graph with obligatory branches $c$ and $g$]{
   \includegraphics[scale =0.7] {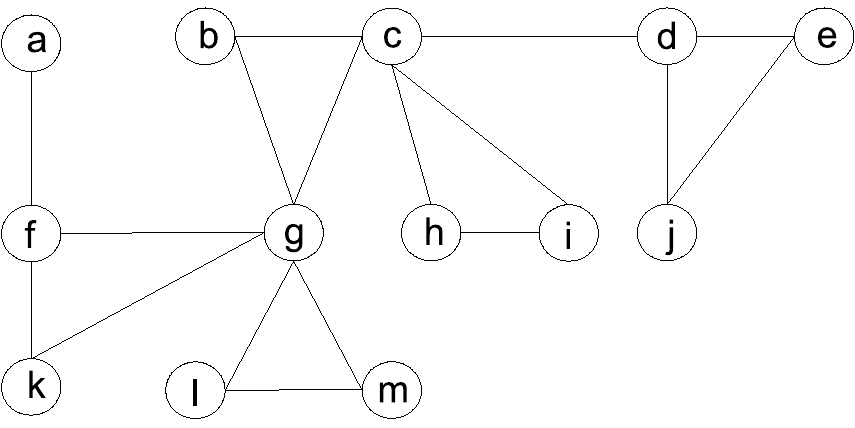}
   \label{fig:decompositionoblbr_1} 
 }
 \hspace{0.7cm}
 \subfigure[Applying decomposition for obligatory branches $c$ and $g$]{
   \includegraphics[scale =0.7] {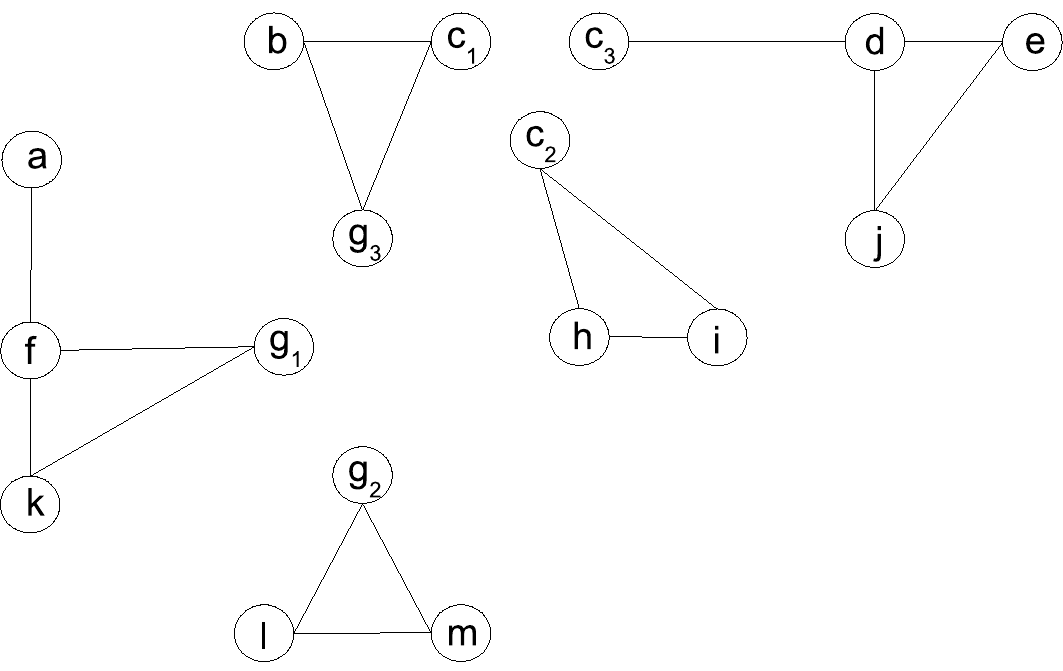}
   \label{fig:decompositionoblbr_2}
 }
\caption{Illustration of the obligatory branches based decomposition}
\label{fig:decompositionoblbr}
\end{figure}

\subsection{Cut edges based decomposition}

This decomposition is based on the fact that every cut edge must be in any spanning tree of $G$.
Starting from a graph $G=(V,E)$, let $B  = \{(u,v) \in E: (u,v) \textrm{ is a cut edge in G}\}$.

A new graph $G'$ can be obtained from $G$ by removing from $E$ every edge $(u,v) \in B$. 
For the sake of correctness of the formulation for each resulting subproblem defined over $G'$ (see Section \ref{sec:decomposed-problem}), we need the information of \emph{extra degree} corresponding to both end vertices of the cut edges removed in $G'$. Let $\gamma(v)$ be the \textit{extra degree} of a vertex $v$, which corresponds to the number of cut edges $(u,v) \in B$ incident to $v$ in the initial graph $G$.

The cut edges based decomposition is illustrated in Figure \ref{fig:decompositioncutedges} for a graph with two cut edges.

 \begin{figure}[!ht]
\centering
\subfigure[Original graph with cut edges $(a,f)$ and $(c,d)$]{
   \includegraphics[scale =0.7] {DecompositionsExample.pdf}
   \label{fig:decompositioncutedges_1} 
 }
 \hspace{0.7cm}
 \subfigure[Applying decomposition for cut edges $(a,f)$ and $(c,d)$, which implies $\gamma(a)=\gamma(f)=\gamma(c)=\gamma(d)=1$]{
   \includegraphics[scale =0.7] {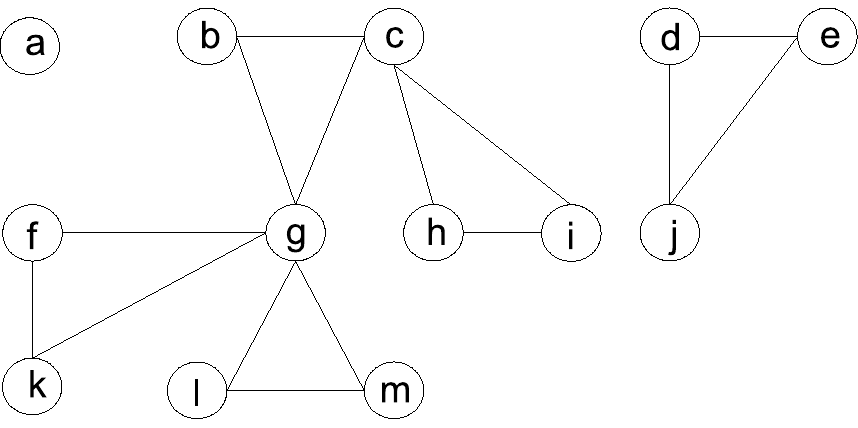}
   \label{fig:decompositioncutedges_2}
 }
\caption{Illustration of the cut edges based decomposition}
\label{fig:decompositioncutedges}
\end{figure}

\subsection{The decomposed problem}
\label{sec:decomposed-problem}

We now give an integer programming formulation aimed at solving each component of the decomposed problem as an independent subproblem.
Let $G'_o=(V'_o,E'_o)$ be the resulting graph after performing the obligatory branches and cut edges based decompositions. Observe that $G'_o$ is disconnected in case at least one obligatory branch or one cut edge was found.

Let $K$ be the number of connected components in $G'_o$, 
and ${G'}_o^k=({V'}_o^{k},{E'}_o^k)$  be the $k-th$ connected subgraph of $G'_o$, with $k \in \{1, \ldots, K \}$. 
Define the set ${\overline{V}'}^{k}_o = \{v \in {V'}^k_o: v \textrm{ was not artificially created based on } L_o\}$.
A formulation for each of the $k$ connected components can be obtained as:

\begin{alignat}{2}
z_k = \min&\quad  \sum_{v \in \overline{V}'^k_o} y_v \label{dec-obj}\\
\text{s.t.}&\quad \nonumber \\
&  \sum_{(u,v)\in {E'}_o^k} x_{uv} = |{V'}^k_o|-1, \quad \label{dec-01} \\
&  \sum_{(u,v)\in {E'}_o^k:u,v \in S} x_{uv} \leq |S|-1  \ \ \ \quad &&  {\rm for \ } S \subset {V'}^k_o,  \label{dec-02}\\
&  \sum_{u:(u,v)\in {E'}_o^k} x_{uv} + \gamma(v) \leq 2 + (d_G(v) -2)y_v \ \ \ \quad &&  {\rm for \ } v \in {\overline{V}'}^k_o,  \label{dec-03}\\
& x \in \{0,1 \}^{|{E'}_o^k|}, \quad \ y \in \{0,1\}^{|{V'}^k_o|}. \label{dec-04}
\end{alignat}

\vspace{0.5cm}

\begin{Prop}
An optimal solution to the minimum branch vertices problem can be obtained from the solutions of the $K$ connected components of $G'_o$ and its optimal value is $$z=|L_o| + \sum_{k=1}^K z_k.$$
\end{Prop}

\begin{proof}
By definition, each obligatory branch should contribute with one unit to the objective function $z$. Since they are not considered in the decomposed objective functions $z_k$, the $|L_o|$ units have to be added to $z$.

Since every cut edge $(u,v) \in B$ must be in the spanning tree, the values $\gamma(u)$ and $\gamma(v)$ of its incident vertices will guarantee that the additional degree corresponding to these edges will be taken into account in ${G'}^k _o$.
Considering an obligatory branch $v$, there are no edges linking vertices between two different components in the original graph without $v$, i.e. $G-v$. Therefore the different components of $G-v$ directly implied by the removal of $v$ have to be connected through the edges incident to $v$ in $G$.

In every connected component ${G'}^k_o$, we obtain a spanning tree minimizing the number of branch vertices, disregarding the obligatory branches implied by $L_o$.
The different trees will be connected to the other components in the original graph via either the obligatory branch vertices or the cut edges.
The resulting forest implies thus a spanning tree with minimum number of branch vertices in the original graph $G$.
\end{proof}

\section{Heuristic algorithms}
\label{sec:heuristics}

We noted that the heuristics available in the literature for the MBV do not take advantage of the problem's structure in order to generate good feasible solutions. Instead, they concentrate into the improvement of available solutions.
We present two simple constructive heuristic algorithms to the MBV 
that take the structure of the problem into consideration in an attempt to 
obtain high quality solutions.

The key observation is the fact that in an ideal situation (one in which no branch vertices are necessary), the optimal spanning tree to the problem should be a branch vertex free solution, i.e. a Hamiltonian path. 
Observe also that good solutions will tend to have paths that are as long as possible connected to each other via some branch vertices. The two constructive heuristics presented in this section take this into consideration.

\subsection{Path expanding heuristic}
The basic idea of the path expanding heuristic is to, starting from a tree $T$ with a single vertex, constructively turn $T$ into a spanning tree by expanding paths already in $T$.
The algorithm has two main components, namely a \textit{start-restart} in which a new vertex is selected as source of a new path and a \textit{path expansion} in which new vertices are added to the path being expanded. The greedy criteria for the two components are:
\begin{itemize}
    \item start-restart (enumerated in order of priority):
    \begin{enumerate}
        \item obligatory branch vertex;
        \item vertex whose degree is already greater than two in the tree;
        \item vertex with the largest number of neighbors which are still not in the tree;
    \end{enumerate}
    \item path expansion: vertex, adjacent to the latest one added to the tree, with the smallest number of neighbors which are still not in the tree.
\end{itemize}

The heuristic is presented in Algorithm \ref{alg:pathexpanding} and it works as follows. Start the tree with a vertex selected according to the start-restart greedy criterion. Build a path starting at a vertex already belonging to the partial tree, selecting the next vertices in this path according to the path expansion criterion until it can no longer be expanded. If there are still vertices not belonging to the tree, another vertex with unvisited neighbors in the partial tree is selected and a new path is constructed starting from the current tree.
The algorithm is illustrated in Figure~\ref{fig:pathexpanding}.

\begin{algorithm}[!ht]
\caption {Path-Expanding($G=(V,E)$)}
\label{alg:pathexpanding}
\begin{algorithmic}[1]
	\STATE {$ V(T) \leftarrow \{u\}$, with $u$ selected according to the start-restart criterion}
	\WHILE {$|T| < |V|-1$}
	    \STATE {select a vertex $u \in V(T)$ with neighbors still not in $V(T)$ and degree at most 1, if any exists; otherwise, select the best vertex according to the start-restart criterion}
	    \WHILE{$u$ has neighbors not belonging to $V(T)$}
	        \STATE {$v \leftarrow $ best neighbor of $u$ according to the path expansion criterion such that $v\notin V(T)$ }
	        \STATE {$V(T) \leftarrow V(T) + \{v\}$}
            \STATE {$T \leftarrow T + \{(u,v)\}$}
	        \STATE {$u \leftarrow v$}
	    \ENDWHILE
	\ENDWHILE
	\RETURN $T$
\end{algorithmic}
\end{algorithm}

 \begin{figure}[!hb]
\centering
\subfigure[]{
   \includegraphics[scale =0.55] {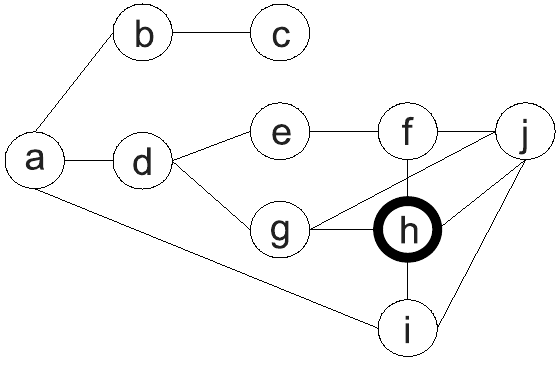}
   \label{fig:pathexpanding_0} 
 }
\subfigure[]{
   \includegraphics[scale =0.55] {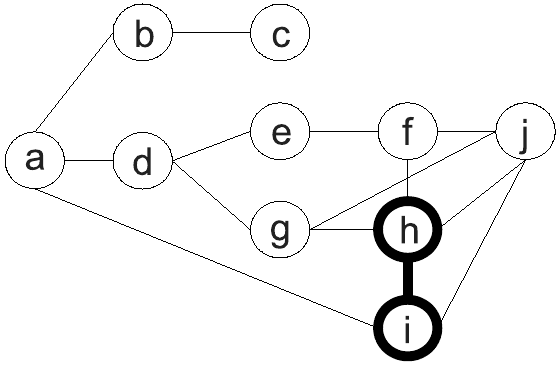}
   \label{fig:pathexpanding_1} 
 }
 \subfigure[]{
   \includegraphics[scale =0.55] {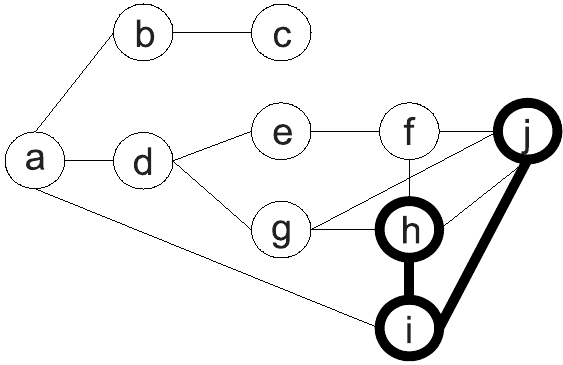}
   \label{fig:pathexpanding_2}
 }
\subfigure[]{
   \includegraphics[scale =0.55] {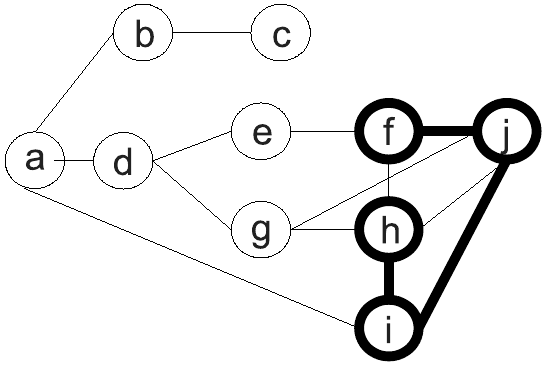}
   \label{fig:pathexpanding_3} 
 }
 \subfigure[]{
   \includegraphics[scale =0.55] {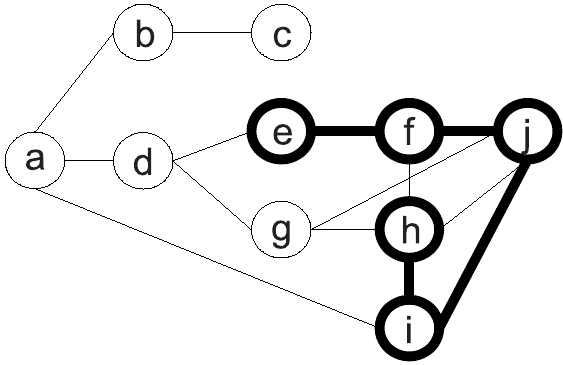}
   \label{fig:pathexpanding_4}
 }\subfigure[]{
   \includegraphics[scale =0.55] {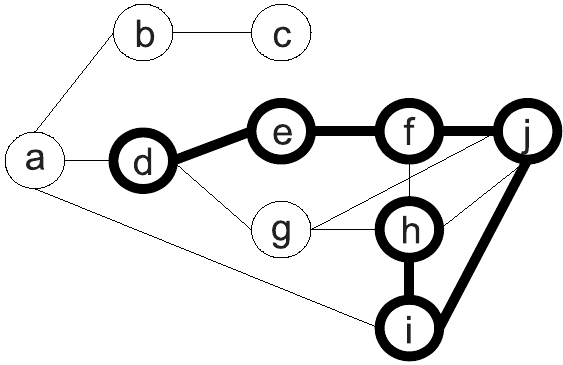}
   \label{fig:pathexpanding_5} 
 }
 \subfigure[]{
   \includegraphics[scale =0.55] {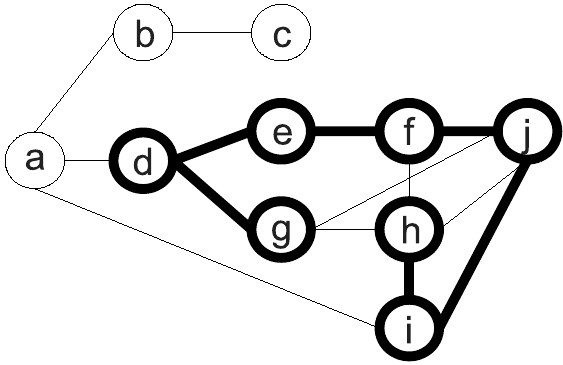}
   \label{fig:pathexpanding_6}
 }\subfigure[]{
   \includegraphics[scale =0.55] {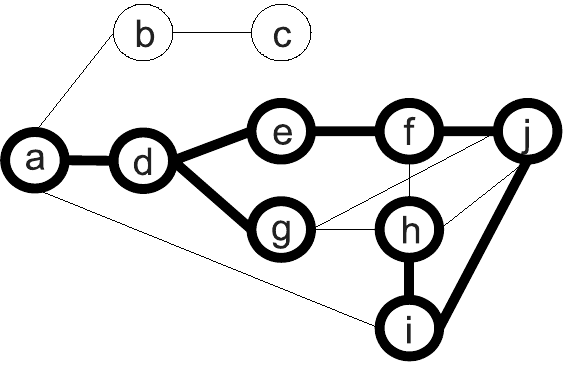}
   \label{fig:pathexpanding_7} 
 }
 \subfigure[]{
   \includegraphics[scale =0.55] {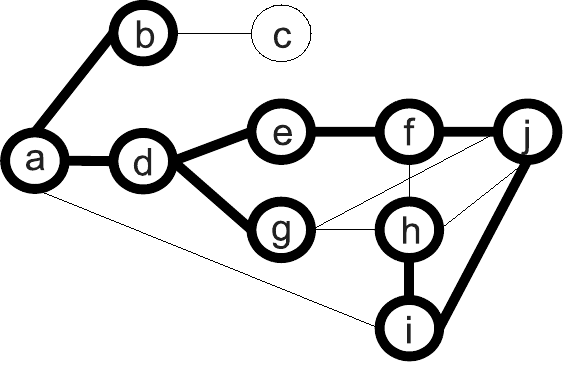}
   \label{fig:pathexpanding_8}
 }\subfigure[]{
   \includegraphics[scale =0.55] {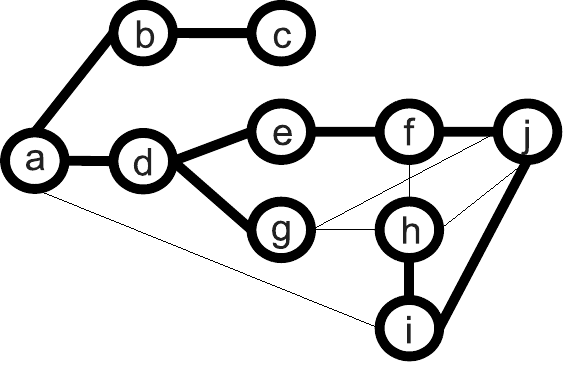}
   \label{fig:pathexpanding_9} 
 }
\caption{Illustration of the path expanding heuristic. The vertices and edges in bold are those belonging to the tree being constructed.}
\label{fig:pathexpanding}
\end{figure}

\begin{Obs}
Algorithm \ref{alg:pathexpanding} runs in $O(|V|^2)$.
\end{Obs}

\begin{proof}
A single evaluation of the start-restart or the path expansion criteria can be done in time proportional to $O(|V|)$, as it requires checking information available at each vertex. In particular, note that checking the number of neighboring vertices which are still not in the tree, \emph{e.g.} the last start-restart criterion and the inner loop condition, can be kept at constant asymptotic complexity, provided that we update this information when the tree is grown. 

The number of iterations of the outer and the inner \textbf{while} loops depends on the particular graph, but they sum to $|V|-1$ since these correspond to the selection of edges in a spanning tree. 
Therefore, the execution includes $|V|-1$ steps, each requiring an $O(|V|)$ evaluation (either start-restart or path expansion) and $O(|V|)$ operations to update the aforementioned information on the remaining vertices.
The algorithm thus runs in time proportional to $O(|V|-1) \times (O(|V|) + O(|V|)) = O(|V|^2)$.
\end{proof}

\subsection{Multi-path expanding heuristic}

The multi-path expanding heuristic differs from the path expanding heuristic in the fact that it allows the expansion of any of the paths being expanded at a given moment.
It also has a start-restart component, which is the same for the path-expanding heuristic. The multi-path expansion component uses the following greedy criterion:
\begin{itemize}
  \item multi-path expansion: vertex, adjacent to any vertex belonging to the tree, with the smallest number of neighbors which are still not in the tree.
\end{itemize}

The multi-path expanding heuristic is presented in Algorithm \ref{alg:multipathexpanding} and it works in the following way. 
Start the tree with a given vertex, selected according to the start-restart greedy criterion, which forms the list of candidates for path expansion. 
Continue selecting a vertex adjacent to one of the candidates for path expansion according to the multi-path expansion greedy criterion, add the new vertex to the list of candidates and remove the current candidate from the list of candidates when appropriate.
If there are no more vertices, which are still not in the tree, adjacent to the candidates and if the tree is still not a spanning tree, select another vertex to add to the list of candidates.
The algorithm is illustrated in Figure~\ref{fig:multipathexpanding}.

\begin{algorithm}[!ht]
\caption {Multi-Path-Expanding($G=(V,E)$)}
\label{alg:multipathexpanding}
\begin{algorithmic}[1]
	\STATE {$ V(T) \leftarrow \{u\}$, with $u$ selected according to the start-restart criterion}
	\STATE{$Cand \leftarrow \varnothing$} 
	\WHILE {$|T|<|V|-1$}
		\STATE{$Cand \leftarrow Cand \cup \{u\}, u \in V(T)$ according to the start-restart criterion}
	    \WHILE{ $\exists (u,v)$ such that $u \in Cand$ and $v \notin V(T)$}
	        \STATE {select such $(u,v)$ optimizing the multi-path expansion criterion}
	        \STATE {$T \leftarrow T + \{(u,v)\}$}
	        \STATE {$V(T) \leftarrow V(T) + \{v\}$}
		\IF {$d_T(u)=2$ and $u \notin L_0$}
			\STATE{$Cand \leftarrow Cand - \{u\}$}
		\ENDIF	        
		\STATE{$Cand \leftarrow Cand +\{v\}$}
	    \ENDWHILE
	\ENDWHILE
	\RETURN $T$
\end{algorithmic}
\end{algorithm}

 \begin{figure}[!ht]
\centering
\subfigure[]{
   \includegraphics[scale =0.55] {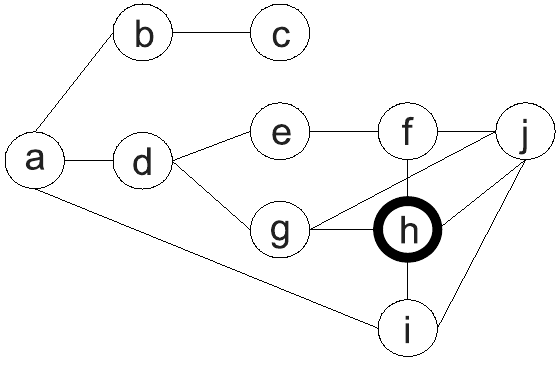}
   \label{fig:multipathexpanding_1} 
 }
 \subfigure[]{
   \includegraphics[scale =0.55] {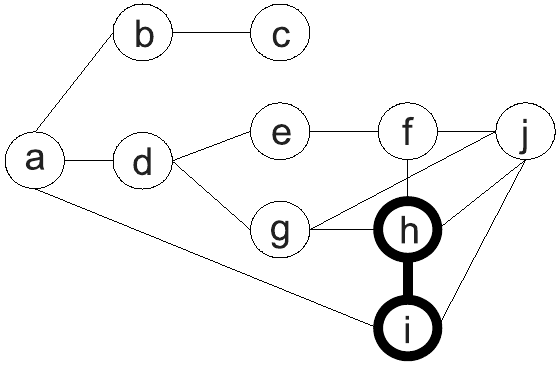}
   \label{fig:multipathexpanding_2}
 }
\subfigure[]{
   \includegraphics[scale =0.55] {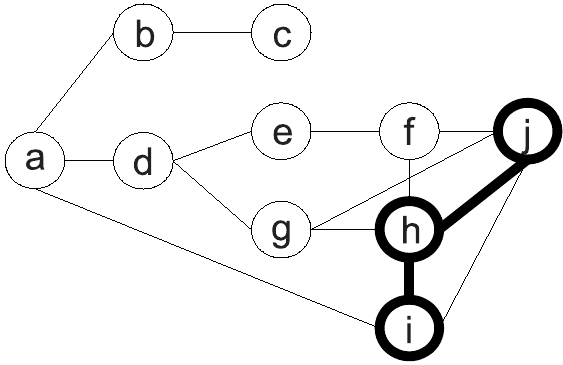}
   \label{fig:multipathexpanding_3} 
 }
 \subfigure[]{
   \includegraphics[scale =0.55] {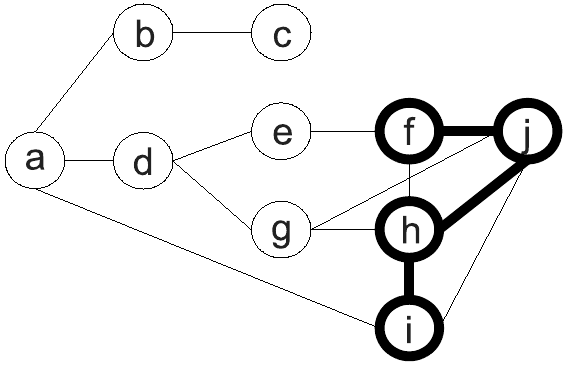}
   \label{fig:multipathexpanding_4}
 }\subfigure[]{
   \includegraphics[scale =0.55] {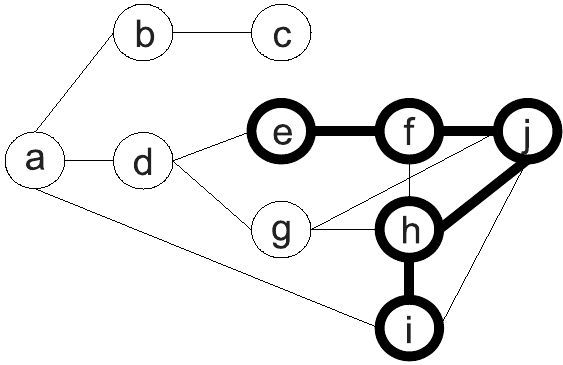}
   \label{fig:multipathexpanding_5} 
 }
 \subfigure[]{
   \includegraphics[scale =0.55] {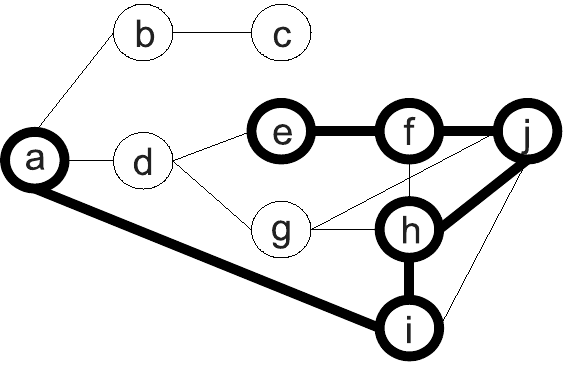}
   \label{fig:multipathexpanding_6}
 }\subfigure[]{
   \includegraphics[scale =0.55] {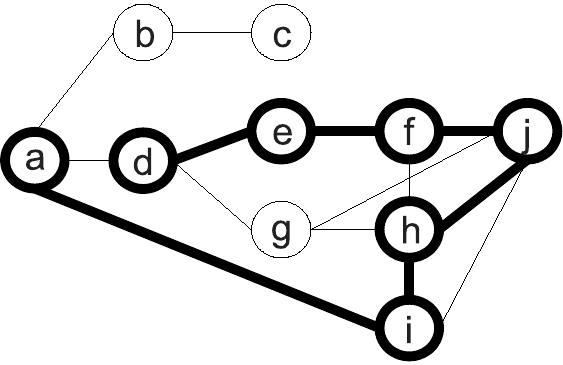}
   \label{fig:multipathexpanding_7} 
 }
 \subfigure[]{
   \includegraphics[scale =0.55] {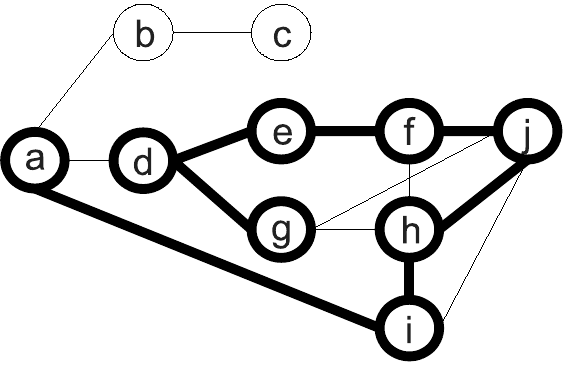}
   \label{fig:multipathexpanding_8}
 }\subfigure[]{
   \includegraphics[scale =0.55] {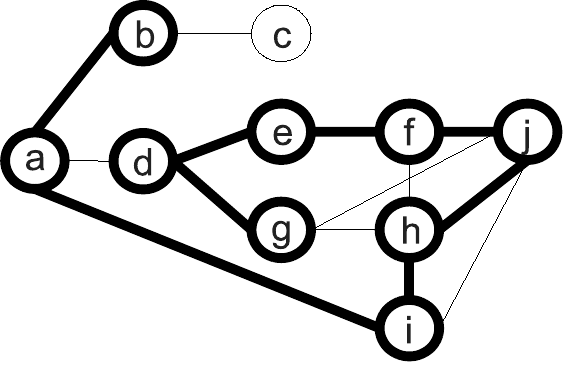}
   \label{fig:multipathexpanding_9} 
 }
 \subfigure[]{
   \includegraphics[scale =0.55] {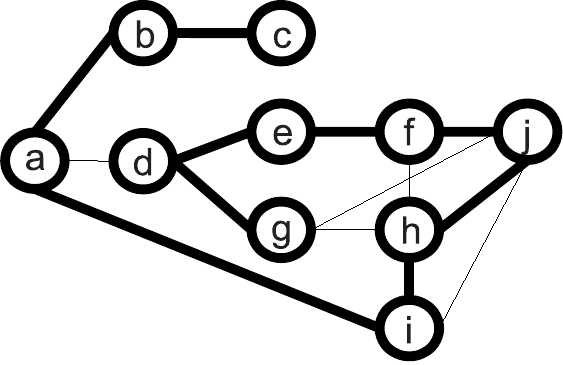}
   \label{fig:multipathexpanding_10} 
 }
\caption{Illustration of the multi-path expanding heuristic. The vertices and edges in bold are those belonging to the tree being constructed.}
\label{fig:multipathexpanding}
\end{figure}

\begin{Obs}
Algorithm \ref{alg:multipathexpanding} runs in $O(|V|^2)$.
\end{Obs}
\begin{proof}
The analysis is analogous to that for Algorithm \ref{alg:pathexpanding}. The main difference is that, to keep the multi-path expansion test in $O(|V|)$ time, we need to store the additional information of whether a given unvisited vertex is adjacent to one in the tree. This can be included in the structure update without increasing its asymptotic complexity. The operations introduced due to the \emph{candidates} set are done in constant order of complexity. 
    
Again, there is a total of $|V| - 1$ iterations considering both \textbf{while} loops, each selecting a new edge in the tree.
When needed, choosing a vertex to enter the set of candidates is done following the same start-restart criterion, in time $O(|V|)$.
In the inner loop, both the multi-path expansion and the updating step require $O(|V|)$ time. Therefore, the algorithm requires computational time proportional to $O(|V|-1) \times (O(|V|) + O(|V|)) = O(|V|^2)$.
\end{proof}

\section{Computational experiments}
\label{sec:computation}

The main goal of our computational experiments is to assess the effectiveness of the methods proposed in this paper.
The branch and cut algorithm presented by Rossi et al.~\cite{Rosetal14} is based 
on the formulation (\ref{obj})--(\ref{const-04}), and it is a 
\emph{state-of-the-art} among exact approaches to solve the MBV~\footnote{A 
very recent paper by Mar\'in~\cite{Mar15} might also be considered as a 
\emph{state-of-the-art}. In that article Rossi et al.~\cite{Rosetal14} is not 
cited and consequently no comparison is done. While the approach in 
Mar\'in~\cite{Mar15} seems to be better for some instances of Carrabs et al.~\cite{Caretal13}, it fails in finding optimal solutions for several instances of Silva et al.~\cite{Siletal14} within one hour of computation with parallel execution. 
Those instances are solved to optimality in this work with the same time limit and with a single thread of execution.}.
We implemented our own branch and cut algorithm using the formulation (\ref{obj})--(\ref{const-04}), and refer to it as the \emph{plain algorithm}.
Our main results compare it against one variant developed on top of the decomposition and primal heuristics we proposed in previous sections: henceforth, the \emph{enhanced algorithm}, using formulation (\ref{dec-obj})--(\ref{dec-04}).

The algorithms are 
implemented in \emph{C++}, using the callback mechanism in the Concert API of 
CPLEX 12.6. 
Since the objective function is integer-valued for all the available instances, we set the parameter regarding the absolute MIP gap tolerance to $0.9999$; all remaining options are used as default.
Also, the best solution among the path-expanding and the multi-path-expanding heuristics is given as a MIP start to the solver.
Experiments were carried out on a machine with an Intel Core i7-4790K (4.00GHz) CPU, with 16GB of RAM. 
A time limit of one hour (3600s wall clock) was imposed for every execution.

We remark that the whole procedure consisting of the decomposition and the constructive methods runs immediately for all the available benchmark instances. Considering any of those instances, the decomposed integer programming (IP) models and heuristic solutions were created in less than $0.02$ seconds. 

Most papers on the MBV use the large benchmark set of Carrabs et al.~\cite{Caretal13}, which includes 525 instances. Section \ref{sec:xp-carrabs} reports results using that benchmark. Nevertheless, Silva et al.~\cite{Siletal14} introduced different sets of instances; for the sake of completeness, we discuss in Section \ref{sec:xp-silva} the effectiveness of the procedures using those problem sets as well.

\subsection{Instances from Carrabs et al.~\cite{Caretal13}}
\label{sec:xp-carrabs}

We describe next our computational results using the standard benchmark of Carrabs et al.~\cite{Caretal13}, which was also used in~\cite{Onc14,Rosetal14}. 
The instances are identified as \texttt{Spd\_RF2\_n\_m\_r}, in which $n$ stands for the number of vertices, $m$ for the number of edges, and $r$ the seed used to randomly generate the edges of the graph.

We start by describing the effect of the decomposition and comparing the plain and enhanced algorithms. 
Table \ref{tab:openInstances} presents the results for all instances for which an optimal solution was not found by the plain algorithm within the one hour time limit.
The second and third columns indicate the number of obligatory branches (OB) and cut edges (CE) 
removed using the decomposition algorithm.
The next three columns concern the enhanced algorithm, and show the best lower and upper bounds and the remaining open gap (in \%) between them at the end of the execution. The last three columns give the same information for the plain algorithm.
We also show the time taken for the enhanced algorithm to solve the instances to optimality (the value 3600 is shown in case the instance remained unsolved at the end of the time limit).
The arithmetic mean is used for average values.

In the comparison, we argue in favor of the enhanced algorithm for two main reasons. First, it is consistently superior in finding better solutions for the MBV.
Note that the enhanced algorithm yields a smaller duality gap in 48 out of 50 instances ($96\%$). In fact, it is able to close the gap of 39 among these instances ($78\%$).

\begin{scriptsize}
\begin{longtable}[t]{@{}lcccccccccccc@{}}
\caption{Results regarding instances for which the plain branch and cut
algorithm cannot prove optimality within one hour of computation. For the starred instances, the enhanced algorithm is worse than the plain version. \label{tab:openInstances}}\\
\toprule
\multirow{2}{*}{Instance}    &    \multicolumn{2}{c}{Reduction}            &    \hspace{5pt}    &    \multicolumn{4}{c}{Enhanced Algorithm}                            &    \hspace{5pt}    &    \multicolumn{3}{c}{Plain Algorithm}                    \\
    &    OB    &    CE    &        &    LB    &    UB    &    Gap (\%)    &    Time (s)    &        &    LB    &    UB    &    Gap    (\%)\\
\midrule
\endfirsthead
\caption[]{\small (continued) Results regarding instances for which the plain branch and cut 
algorithm cannot prove optimality within one hour of computation. For the starred instances, the enhanced algorithm is worse than the plain version.}\\
\toprule
\multirow{2}{*}{Instance}	&	\multicolumn{2}{c}{Reduction}			&	\hspace{5pt}	&	\multicolumn{4}{c}{Enhanced Algorithm}							&	\hspace{5pt}	&	\multicolumn{3}{c}{Plain Algorithm}					\\
	&	OB	&	CE	&		&	LB	&	UB	&	Gap (\%)	&	Time (s)	&		&	LB	&	UB	&	Gap	(\%)\\
\midrule
\endhead
\texttt{Spd\_RF2\_400\_519\_4731}   &   52  &   155 &       &   70  &   70  &   0   &   227.9   &       &    67.66  &   70  &   3.3 \\
\texttt{Spd\_RF2\_450\_548\_4915}   &   68  &   205 &       &   89  &   89  &   0   &   324.4   &       &    88.00  &   89  &   1.1 \\
\texttt{Spd\_RF2\_450\_581\_4947}   &   59  &   178 &       &   76.38   &   80  &   4.5 &   3600.0  &       &    75.86  &   80  &   5.2 \\
\texttt{Spd\_RF2\_450\_581\_4963}   &   61  &   178 &       &   77  &   77  &   0   &   125.8   &       &    75.81  &   78  &   2.8 \\
\texttt{Spd\_RF2\_450\_614\_4979*}  &   45  &   149 &       &   64.35   &   67  &   4.0 &   3600.0  &       &    64.34  &   66  &   2.5 \\
\texttt{Spd\_RF2\_450\_614\_5003}   &   44  &   153 &       &   67  &   67  &   0   &   901.8   &       &    64.55  &   68  &   5.1 \\
\texttt{Spd\_RF2\_500\_603\_5091}   &   90  &   264 &       &   109 &   109 &   0   &   176.1   &       &    107.37 &   109 &   1.5 \\
\texttt{Spd\_RF2\_500\_672\_5171}   &   58  &   180 &       &   79.72   &   82  &   2.8 &   3600.0  &       &    78.24  &   81  &   3.4 \\
\texttt{Spd\_RF2\_500\_672\_5179}   &   52  &   171 &       &   75.45   &   77  &   2.0 &   3600.0  &       &    73.83  &   78  &   5.3 \\
\texttt{Spd\_RF2\_500\_672\_5187*}  &   47  &   155 &       &   69.51   &   72  &   3.5 &   3600.0  &       &    69.53  &   71  &   2.1 \\
\texttt{Spd\_RF2\_500\_672\_5195}   &   57  &   171 &       &   77  &   77  &   0   &   2128.6  &       &    75.69  &   78  &   3.0 \\
\texttt{Spd\_RF2\_500\_672\_5203}   &   57  &   173 &       &   74.58   &   76  &   1.9 &   3600.0  &       &    73.47  &   78  &   5.8 \\
\texttt{Spd\_RF2\_600\_749\_5355}   &   136 &   370 &       &   143 &   143 &   0   &   462.7   &       &    142.00 &   143 &   0.7 \\
\texttt{Spd\_RF2\_700\_821\_5523}   &   164 &   467 &       &   182 &   182 &   0   &   12.6    &       &    180.71 &   183 &   1.3 \\
\texttt{Spd\_RF2\_700\_861\_5563}   &   151 &   437 &       &   159.67  &   161 &   0.8 &   3600.0  &       &    159.72 &   161 &   0.8 \\
\texttt{Spd\_RF2\_700\_902\_5579}   &   143 &   408 &       &   152 &   152 &   0   &   1106.3  &       &    150.73 &   153 &   1.5 \\
\texttt{Spd\_RF2\_700\_902\_5595}   &   143 &   401 &       &   153 &   153 &   0   &   130.7   &       &    151.86 &   154 &   1.4 \\
\texttt{Spd\_RF2\_800\_886\_5659}   &   212 &   601 &       &   227 &   227 &   0   &   10.5    &       &    225.67 &   228 &   1.0 \\
\texttt{Spd\_RF2\_800\_930\_5715}   &   200 &   549 &       &   212 &   212 &   0   &   49.0    &       &    209.42 &   213 &   1.7 \\
\texttt{Spd\_RF2\_800\_973\_5747}   &   183 &   506 &       &   197 &   197 &   0   &   3461.6  &       &    194.36 &   198 &   1.8 \\
\texttt{Spd\_RF2\_800\_973\_5763}   &   179 &   504 &       &   195 &   195 &   0   &   1057.0  &       &    191.11 &   197 &   3.0 \\
\texttt{Spd\_RF2\_800\_1017\_5771}  &   164 &   470 &       &   176 &   176 &   0   &   1208.5  &       &    173.50 &   178 &   2.5 \\
\texttt{Spd\_RF2\_800\_1017\_5779}  &   158 &   467 &       &   173 &   173 &   0   &   2375.2  &       &    171.25 &   173 &   1.0 \\
\texttt{Spd\_RF2\_800\_1017\_5787}  &   165 &   468 &       &   178 &   178 &   0   &   108.6   &       &    175.71 &   181 &   2.9 \\
\texttt{Spd\_RF2\_900\_989\_5859}   &   238 &   688 &       &   259 &   259 &   0   &   63.7    &       &    254.08 &   259 &   1.9 \\
\texttt{Spd\_RF2\_900\_1034\_5891}  &   228 &   640 &       &   242 &   242 &   0   &   23.9    &       &    238.78 &   243 &   1.7 \\
\texttt{Spd\_RF2\_900\_1034\_5907}  &   220 &   621 &       &   239 &   239 &   0   &   1411.3  &       &    236.19 &   240 &   1.6 \\
\texttt{Spd\_RF2\_900\_1034\_5915}  &   228 &   642 &       &   242 &   242 &   0   &   466.1   &       &    238.99 &   243 &   1.7 \\
\texttt{Spd\_RF2\_900\_1079\_5931}  &   206 &   579 &       &   222 &   222 &   0   &   332.6   &       &    218.77 &   223 &   1.9 \\
\texttt{Spd\_RF2\_900\_1079\_5939}  &   207 &   580 &       &   224 &   224 &   0   &   349.5   &       &    220.67 &   226 &   2.4 \\
\texttt{Spd\_RF2\_900\_1079\_5947}  &   210 &   582 &       &   226 &   226 &   0   &   888.0   &       &    220.81 &   230 &   4.0 \\
\texttt{Spd\_RF2\_900\_1079\_5963}  &   205 &   588 &       &   222 &   222 &   0   &   27.9    &       &    218.52 &   223 &   2.0 \\
\texttt{Spd\_RF2\_900\_1124\_5971}  &   198 &   543 &       &   210 &   210 &   0   &   836.7   &       &    207.86 &   210 &   1.0 \\
\texttt{Spd\_RF2\_900\_1124\_5979}  &   198 &   549 &       &   210 &   210 &   0   &   167.9   &       &    207.16 &   210 &   1.4 \\
\texttt{Spd\_RF2\_900\_1124\_5987}  &   190 &   546 &       &   203 &   203 &   0   &   340.6   &       &    201.54 &   204 &   1.2 \\
\texttt{Spd\_RF2\_900\_1124\_6003}  &   189 &   549 &       &   202 &   202 &   0   &   3120.8  &       &    199.06 &   203 &   1.9 \\
\texttt{Spd\_RF2\_1000\_1143\_6091} &   252 &   704 &       &   272 &   272 &   0   &   2334.6  &       &    268.02 &   274 &   2.2 \\
\texttt{Spd\_RF2\_1000\_1143\_6099} &   254 &   711 &       &   271 &   271 &   0   &   9.5 &       &    268.24 &   272 &   1.4 \\
\texttt{Spd\_RF2\_1000\_1143\_6107} &   251 &   698 &       &   272 &   272 &   0   &   108.1   &       &    267.17 &   273 &   2.1 \\
\texttt{Spd\_RF2\_1000\_1143\_6115} &   253 &   700 &       &   272 &   272 &   0   &   314.5   &       &    269.87 &   272 &   0.8 \\
\texttt{Spd\_RF2\_1000\_1143\_6123} &   251 &   712 &       &   269 &   269 &   0   &   906.0   &       &    268.40 &   270 &   0.6 \\
\texttt{Spd\_RF2\_1000\_1191\_6131} &   228 &   662 &       &   247.67  &   249 &   0.5 &   3600.0  &       &    243.83 &   250 &   2.5 \\
\texttt{Spd\_RF2\_1000\_1191\_6139} &   235 &   659 &       &   251 &   251 &   0   &   362.4   &       &    248.11 &   252 &   1.5 \\
\texttt{Spd\_RF2\_1000\_1191\_6147} &   233 &   656 &       &   250 &   250 &   0   &   945.6   &       &    246.47 &   251 &   1.8 \\
\texttt{Spd\_RF2\_1000\_1191\_6155} &   235 &   655 &       &   253 &   253 &   0   &   193.9   &       &    251.54 &   253 &   0.6 \\
\texttt{Spd\_RF2\_1000\_1239\_6171} &   218 &   608 &       &   232 &   232 &   0   &   90.9    &       &    230.11 &   232 &   0.8 \\
\texttt{Spd\_RF2\_1000\_1239\_6179} &   222 &   605 &       &   236 &   237 &   0.4 &   3600.0  &       &    231.26 &   239 &   3.2 \\
\texttt{Spd\_RF2\_1000\_1239\_6187} &   217 &   611 &       &   229.81  &   233 &   1.4 &   3600.0  &       &    225.33 &   237 &   4.9 \\
\texttt{Spd\_RF2\_1000\_1239\_6195} &   222 &   616 &       &   238.01  &   241 &   1.2 &   3600.0  &       &    229.71 &   243 &   5.5 \\
\texttt{Spd\_RF2\_1000\_1239\_6203} &   218 &   609 &       &   236 &   236 &   0   &   2850.6  &       &    230.84 &   237 &   2.6 \\
\midrule
Average	&	20.4\%	&	48.1\%	&		&		&	 	&	 0.46	&		&		&	 	&		&	 2.3	\\
\bottomrule
\end{longtable}
\end{scriptsize}

For all the 475 remaining cases in this benchmark, the enhanced algorithm makes finding optimal solutions for the MBV a much faster task, as we present in Table \ref{tab:solvedInstances}.
Given the large number of instances, we present the arithmetic mean after aggregating the results for instances with the same number of vertices (\emph{n}) and edges (\emph{m}).
Regarding the smaller instances set (up to 500 vertices), the decomposition approach is able to reduce 15.9\% of the vertices and 38.6\% of the edges, on average, thus yielding the optimal solution in a smaller amount of time.
As for the larger instances set (from 600 to 1,000 vertices), the average reduction is of 24.8\% of the vertices and 61.3\% of the edges, achieving a solution much faster.

\begin{scriptsize}
\begin{longtable}[t]{@{}cccccc|cccccc@{}}
\caption{Results regarding instances for which both the plain branch and cut algorithm and the enhanced one are able to prove optimality within one hour of computation. The \emph{enhanced algorithm} features the decomposition and heuristic methods. Columns OB and CE indicate the number of obligatory branches and cut edges removed, respectively. \label{tab:solvedInstances}}\\
\toprule
\multicolumn{2}{c}{Instances}			&	\multicolumn{2}{c}{Reduction Avg.}			&	\multicolumn{2}{c|}{Time Avg. (s) }			&	\multicolumn{2}{c}{Instances}			&	\multicolumn{2}{c}{Reduction Avg.}			&	\multicolumn{2}{c}{Time Avg. (s)}			\\
n	&	m	&	OB	&	CE	&	Enhanced	&	Plain	&	n	&	m	&	OB	&	CE	&	Enhanced	&	Plain	\\
\midrule
\endfirsthead
\caption[]{\small (continued) Results regarding instances for which both the plain branch and cut algorithm and the enhanced one are able to prove optimality within one hour of computation. The \emph{enhanced algorithm} features the decomposition and heuristic methods. Columns OB and CE indicate the number of obligatory branches and cut edges removed, respectively.}\\
\toprule
\multicolumn{2}{c}{Instances}			&	\multicolumn{2}{c}{Reduction Avg.}			&	\multicolumn{2}{c|}{Time Avg. (s) }			&	\multicolumn{2}{c}{Instances}			&	\multicolumn{2}{c}{Reduction Avg.}			&	\multicolumn{2}{c}{Time Avg. (s)}			\\
n	&	m	&	OB	&	CE	&	Enhanced	&	Plain	&	n	&	m	&	OB	&	CE	&	Enhanced	&	Plain	\\
\midrule
\endhead
200	&	222	&	45.8	&	127.8	&	0.1	&	1.5	&	600	&	637	&	173.0	&	493.6	&	0.8	&	26.7	\\
200	&	244	&	32.2	&	92.4	&	1.1	&	4.4	&	600	&	674	&	156.2	&	437.4	&	5.3	&	131.9	\\
200	&	267	&	23.4	&	69.0	&	1.6	&	10.0	&	600	&	712	&	139.6	&	394.4	&	143.0	&	406.6	\\
200	&	289	&	15.4	&	56.8	&	78.5	&	77.2	&	600	&	749	&	129.8	&	361.8	&	55.9	&	837.8	\\
200	&	312	&	11.0	&	42.2	&	11.5	&	15.2	&	600	&	787	&	119.4	&	333.6	&	202.9	&	849.3	\\
250	&	273	&	60.0	&	164.4	&	0.1	&	2.8	&	700	&	740	&	203.0	&	576.8	&	0.7	&	52.0	\\
250	&	297	&	43.2	&	120.8	&	3.0	&	14.9	&	700	&	780	&	185.2	&	518.4	&	9.3	&	242.5	\\
250	&	321	&	34.2	&	101.8	&	3.2	&	14.5	&	700	&	821	&	167.3	&	471.0	&	54.8	&	1285.6	\\
250	&	345	&	24.4	&	76.2	&	42.9	&	76.5	&	700	&	861	&	154.8	&	436.5	&	758.4	&	241.9	\\
250	&	369	&	16.8	&	60.0	&	63.6	&	59.2	&	700	&	902	&	147.0	&	402.3	&	25.0	&	773.5	\\
300	&	326	&	73.2	&	203.0	&	0.3	&	5.3	&	800	&	1017	&	167.0	&	468.0	&	222.2	&	325.2	\\
300	&	353	&	57.8	&	160.2	&	3.8	&	17.8	&	800	&	843	&	232.0	&	666.8	&	1.2	&	69.4	\\
300	&	380	&	43.0	&	124.8	&	20.1	&	57.4	&	800	&	886	&	213.8	&	599.0	&	5.4	&	503.8	\\
300	&	407	&	34.4	&	104.6	&	60.0	&	117.2	&	800	&	930	&	193.5	&	546.0	&	8.4	&	343.1	\\
300	&	434	&	25.6	&	86.4	&	279.0	&	146.2	&	800	&	973	&	180.0	&	506.3	&	1183.7	&	2598.0	\\
350	&	378	&	85.4	&	238.8	&	0.6	&	7.9	&	900	&	1034	&	222.5	&	631.0	&	54.4	&	616.9	\\
350	&	406	&	67.4	&	190.0	&	12.3	&	48.2	&	900	&	1079	&	207.0	&	589.0	&	171.9	&	775.6	\\
350	&	435	&	50.4	&	151.0	&	359.2	&	402.0	&	900	&	1124	&	197.0	&	551.0	&	76.0	&	870.5	\\
350	&	463	&	40.2	&	124.2	&	514.6	&	1059.5	&	900	&	944	&	262.2	&	756.4	&	1.7	&	111.5	\\
350	&	492	&	29.0	&	103.5	&	73.3	&	262.5	&	900	&	989	&	242.0	&	685.0	&	36.6	&	941.7	\\
400	&	429	&	102.4	&	282.6	&	0.8	&	39.4	&	1000	&	1047	&	296.2	&	849.6	&	3.8	&	222.3	\\
400	&	459	&	78.6	&	226.4	&	11.1	&	53.2	&	1000	&	1095	&	268.8	&	767.0	&	32.8	&	820.0	\\
400	&	489	&	62.2	&	184.8	&	44.0	&	228.4	&	1000	&	1191	&	233.0	&	656.0	&	634.6	&	1524.8	\\
400	&	519	&	51.0	&	154.0	&	322.6	&	1198.5	&		&		&		&		&		&		\\
400	&	549	&	41.8	&	131.2	&	700.5	&	1058.8	&		&		&		&		&		&		\\
450	&	482	&	115.0	&	318.6	&	0.9	&	44.5	&		&		&		&		&		&		\\
450	&	515	&	91.6	&	250.6	&	17.4	&	308.5	&		&		&		&		&		&		\\
450	&	548	&	74.0	&	209.8	&	44.3	&	442.5	&		&		&		&		&		&		\\
450	&	581	&	62.0	&	177.5	&	452.8	&	594.3	&		&		&		&		&		&		\\
450	&	614	&	47.7	&	152.3	&	909.5	&	1278.0	&		&		&		&		&		&		\\
500	&	534	&	131.8	&	361.0	&	1.5	&	46.7	&		&		&		&		&		&		\\
500	&	568	&	106.0	&	294.2	&	26.8	&	256.4	&		&		&		&		&		&		\\
500	&	603	&	86.0	&	241.5	&	318.5	&	1404.5	&		&		&		&		&		&		\\
500	&	637	&	71.8	&	210.6	&	556.6	&	2103.1	&		&		&		&		&		&		\\
\midrule
\multicolumn{2}{c}{Average:}	&	15.9\%	&	38.6\%	&	145.2	&	337.0	& \multicolumn{2}{c}{Average:}	&	24.8\%	&	61.3\%	&	160.4	&	633.5 \\
\bottomrule
\end{longtable}
\end{scriptsize}

Finally, we compare the performance of our heuristics with the best primal feasible solution cost obtained by the Lagrangian heuristics of Carrabs et al.~\cite{Caretal13}. The authors compiled their best bounds on a set of 175 problem instances.
Table \ref{tab:heuristicsCarrabs} compares those values with the ones provided by our heuristics; we report the best value among the path-expanding and the multi-path expanding algorithms.

We highlight that our proposed constructive algorithms provide better warm starts. In fact, 
they provide better primal bounds to 103 instances. On the other hand, the best result presented by Carrabs et al.~\cite{Caretal13} is better in 54 instances, and is equal to the solution obtained with our constructive algorithms in 18 instances.

\newpage

\begin{scriptsize}
\begin{longtable}{@{}ccc|ccc@{}}
\caption{\small Best primal solution constructed by our proposed algorithms and by the Lagrangian heuristic of Carrabs et al.~\cite{Caretal13}. Values in bold indicate the 103 instances (out of 175) in which our heuristics provide a better solution. \label{tab:heuristicsCarrabs}}\\
\toprule
\multirow{2}{*}{Instance}    &    Best constructed    &    Best result    &    \multirow{2}{*}{Instance}    &    Best constructed    &    Best result    \\
    &    solution    &    by \cite{Caretal13}    &        &    solution    &    by \cite{Caretal13}    \\
\midrule
\endfirsthead
\caption[]{\small (continued) Best primal solution constructed by our proposed algorithms and by the Lagrangian heuristic of Carrabs et al.~\cite{Caretal13}. Values in bold indicate the 103 instances (out of 175) in which our heuristics provide a better solution.}\\
\toprule
\multirow{2}{*}{Instance}    &    Best constructed    &    Best result    &    \multirow{2}{*}{Instance}    &    Best constructed    &    Best result    \\
    &    solution    &    by \cite{Caretal13}    &        &    solution    &    by \cite{Caretal13}    \\
\midrule
\endhead
\texttt{Spd\_RF2\_200\_222\_3811}	&	55	&	54	&	  \texttt{Spd\_RF2\_350\_435\_4515}	&	81	&	77	\\
\texttt{Spd\_RF2\_200\_222\_3819}	&	55	&	54	&	  \texttt{Spd\_RF2\_350\_435\_4523}	&	75	&	74	\\
\texttt{Spd\_RF2\_200\_222\_3827}	&	53	&	51	&	  \texttt{Spd\_RF2\_350\_463\_4531}	&	\textbf{69}	&	70	\\
\texttt{Spd\_RF2\_200\_222\_3835}	&	54	&	52	&	  \texttt{Spd\_RF2\_350\_463\_4539}	&	70	&	70	\\
\texttt{Spd\_RF2\_200\_222\_3843}	&	55	&	54	&	  \texttt{Spd\_RF2\_350\_463\_4547}	&	\textbf{67}	&	71	\\
\texttt{Spd\_RF2\_200\_244\_3851}	&	\textbf{43}	&	46	&	  \texttt{Spd\_RF2\_350\_463\_4555}	&	\textbf{68}	&	70	\\
\texttt{Spd\_RF2\_200\_244\_3859}	&	\textbf{46}	&	49	&	  \texttt{Spd\_RF2\_350\_463\_4563}	&	\textbf{66}	&	73	\\
\texttt{Spd\_RF2\_200\_244\_3867}	&	45	&	43	&	  \texttt{Spd\_RF2\_350\_492\_4571}	&	\textbf{55}	&	62	\\
\texttt{Spd\_RF2\_200\_244\_3875}	&	\textbf{42}	&	43	&	  \texttt{Spd\_RF2\_350\_492\_4579}	&	\textbf{62}	&	64	\\
\texttt{Spd\_RF2\_200\_244\_3883}	&	45	&	45	&	  \texttt{Spd\_RF2\_350\_492\_4587}	&	\textbf{59}	&	60	\\
\texttt{Spd\_RF2\_200\_267\_3891}	&	\textbf{36}	&	37	&	  \texttt{Spd\_RF2\_350\_492\_4595}	&	59	&	57	\\
\texttt{Spd\_RF2\_200\_267\_3899}	&	\textbf{37}	&	39	&	  \texttt{Spd\_RF2\_350\_492\_4603}	&	\textbf{59}	&	63	\\
\texttt{Spd\_RF2\_200\_267\_3907}	&	37	&	37	&	  \texttt{Spd\_RF2\_400\_429\_4611}	&	118	&	118	\\
\texttt{Spd\_RF2\_200\_267\_3915}	&	\textbf{36}	&	39	&	  \texttt{Spd\_RF2\_400\_429\_4619}	&	118	&	116	\\
\texttt{Spd\_RF2\_200\_267\_3923}	&	\textbf{36}	&	41	&	  \texttt{Spd\_RF2\_400\_429\_4627}	&	\textbf{115}	&	116	\\
\texttt{Spd\_RF2\_200\_289\_3931}	&	\textbf{28}	&	30	&	  \texttt{Spd\_RF2\_400\_429\_4635}	&	117	&	115	\\
\texttt{Spd\_RF2\_200\_289\_3939}	&	\textbf{29}	&	35	&	  \texttt{Spd\_RF2\_400\_429\_4643}	&	118	&	117	\\
\texttt{Spd\_RF2\_200\_289\_3947}	&	\textbf{31}	&	35	&	  \texttt{Spd\_RF2\_400\_459\_4651}	&	106	&	105	\\
\texttt{Spd\_RF2\_200\_289\_3955}	&	\textbf{31}	&	34	&	  \texttt{Spd\_RF2\_400\_459\_4659}	&	101	&	100	\\
\texttt{Spd\_RF2\_200\_289\_3963}	&	\textbf{30}	&	32	&	  \texttt{Spd\_RF2\_400\_459\_4667}	&	\textbf{102}	&	106	\\
\texttt{Spd\_RF2\_200\_312\_3971}	&	24	&	23	&	  \texttt{Spd\_RF2\_400\_459\_4675}	&	\textbf{104}	&	106	\\
\texttt{Spd\_RF2\_200\_312\_3979}	&	\textbf{23}	&	27	&	  \texttt{Spd\_RF2\_400\_459\_4683}	&	103	&	103	\\
\texttt{Spd\_RF2\_200\_312\_3987}	&	\textbf{26}	&	29	&	  \texttt{Spd\_RF2\_400\_489\_4691}	&	87	&	87	\\
\texttt{Spd\_RF2\_200\_312\_3995}	&	\textbf{23}	&	26	&	  \texttt{Spd\_RF2\_400\_489\_4699}	&	\textbf{92}	&	98	\\
\texttt{Spd\_RF2\_200\_312\_4003}	&	\textbf{28}	&	30	&	  \texttt{Spd\_RF2\_400\_489\_4707}	&	\textbf{90}	&	95	\\
\texttt{Spd\_RF2\_250\_273\_4011}	&	71	&	71	&	  \texttt{Spd\_RF2\_400\_489\_4715}	&	91	&	88	\\
\texttt{Spd\_RF2\_250\_273\_4019}	&	70	&	66	&	  \texttt{Spd\_RF2\_400\_489\_4723}	&	92	&	92	\\
\texttt{Spd\_RF2\_250\_273\_4027}	&	70	&	69	&	  \texttt{Spd\_RF2\_400\_519\_4731}	&	\textbf{79}	&	83	\\
\texttt{Spd\_RF2\_250\_273\_4035}	&	71	&	69	&	  \texttt{Spd\_RF2\_400\_519\_4739}	&	\textbf{77}	&	85	\\
\texttt{Spd\_RF2\_250\_273\_4043}	&	\textbf{67}	&	70	&	  \texttt{Spd\_RF2\_400\_519\_4747}	&	\textbf{79}	&	82	\\
\texttt{Spd\_RF2\_250\_297\_4051}	&	60	&	60	&	  \texttt{Spd\_RF2\_400\_519\_4755}	&	\textbf{80}	&	85	\\
\texttt{Spd\_RF2\_250\_297\_4059}	&	\textbf{61}	&	63	&	  \texttt{Spd\_RF2\_400\_519\_4763}	&	\textbf{78}	&	84	\\
\texttt{Spd\_RF2\_250\_297\_4067}	&	\textbf{57}	&	56	&	  \texttt{Spd\_RF2\_400\_549\_4771}	&	\textbf{68}	&	77	\\
\texttt{Spd\_RF2\_250\_297\_4075}	&	\textbf{59}	&	63	&	  \texttt{Spd\_RF2\_400\_549\_4779}	&	\textbf{66}	&	73	\\
\texttt{Spd\_RF2\_250\_297\_4083}	&	61	&	60	&	  \texttt{Spd\_RF2\_400\_549\_4787}	&	72	&	71	\\
\texttt{Spd\_RF2\_250\_321\_4091}	&	\textbf{52}	&	57	&	  \texttt{Spd\_RF2\_400\_549\_4795}	&	71	&	69	\\
\texttt{Spd\_RF2\_250\_321\_4099}	&	\textbf{50}	&	53	&	  \texttt{Spd\_RF2\_400\_549\_4803}	&	\textbf{73}	&	75	\\
\texttt{Spd\_RF2\_250\_321\_4107}	&	\textbf{49}	&	51	&	  \texttt{Spd\_RF2\_450\_482\_4811}	&	131	&	128	\\
\texttt{Spd\_RF2\_250\_321\_4115}	&	\textbf{48}	&	54	&	  \texttt{Spd\_RF2\_450\_482\_4819}	&	130	&	130	\\
\texttt{Spd\_RF2\_250\_321\_4123}	&	\textbf{50}	&	54	&	  \texttt{Spd\_RF2\_450\_482\_4827}	&	133	&	133	\\
\texttt{Spd\_RF2\_250\_345\_4131}	&	\textbf{39}	&	41	&	  \texttt{Spd\_RF2\_450\_482\_4835}	&	\textbf{132}	&	133	\\
\texttt{Spd\_RF2\_250\_345\_4139}	&	48	&	47	&	  \texttt{Spd\_RF2\_450\_482\_4843}	&	132	&	130	\\
\texttt{Spd\_RF2\_250\_345\_4147}	&	45	&	44	&	  \texttt{Spd\_RF2\_450\_515\_4851}	&	\textbf{116}	&	118	\\
\texttt{Spd\_RF2\_250\_345\_4155}	&	\textbf{44}	&	45	&	  \texttt{Spd\_RF2\_450\_515\_4859}	&	118	&	118	\\
\texttt{Spd\_RF2\_250\_345\_4163}	&	\textbf{41}	&	47	&	  \texttt{Spd\_RF2\_450\_515\_4867}	&	\textbf{112}	&	120	\\
\texttt{Spd\_RF2\_250\_369\_4171}	&	38	&	38	&	  \texttt{Spd\_RF2\_450\_515\_4875}	&	118	&	118	\\
\texttt{Spd\_RF2\_250\_369\_4179}	&	\textbf{33}	&	35	&	  \texttt{Spd\_RF2\_450\_515\_4883}	&	117	&	114	\\
\texttt{Spd\_RF2\_250\_369\_4187}	&	\textbf{35}	&	37	&	  \texttt{Spd\_RF2\_450\_548\_4891}	&	109	&	105	\\
\texttt{Spd\_RF2\_250\_369\_4195}	&	\textbf{35}	&	38	&	  \texttt{Spd\_RF2\_450\_548\_4899}	&	\textbf{101}	&	105	\\
\texttt{Spd\_RF2\_250\_369\_4203}	&	\textbf{36}	&	40	&	  \texttt{Spd\_RF2\_450\_548\_4907}	&	\textbf{100}	&	101	\\
\texttt{Spd\_RF2\_300\_326\_4211}	&	86	&	85	&	  \texttt{Spd\_RF2\_450\_548\_4915}	&	\textbf{101}	&	106	\\
\texttt{Spd\_RF2\_300\_326\_4219}	&	87	&	83	&	  \texttt{Spd\_RF2\_450\_548\_4923}	&	\textbf{105}	&	107	\\
\texttt{Spd\_RF2\_300\_326\_4227}	&	87	&	85	&	  \texttt{Spd\_RF2\_450\_581\_4931}	&	93	&	93	\\
\texttt{Spd\_RF2\_300\_326\_4235}	&	85	&	85	&	  \texttt{Spd\_RF2\_450\_581\_4939}	&	\textbf{90}	&	93	\\
\texttt{Spd\_RF2\_300\_326\_4243}	&	85	&	84	&	  \texttt{Spd\_RF2\_450\_581\_4947}	&	\textbf{95}	&	99	\\
\texttt{Spd\_RF2\_300\_353\_4251}	&	75	&	73	&	  \texttt{Spd\_RF2\_450\_581\_4955}	&	\textbf{91}	&	97	\\
\texttt{Spd\_RF2\_300\_353\_4259}	&	73	&	72	&	  \texttt{Spd\_RF2\_450\_581\_4963}	&	\textbf{90}	&	96	\\
\texttt{Spd\_RF2\_300\_353\_4267}	&	77	&	77	&	  \texttt{Spd\_RF2\_450\_614\_4971}	&	\textbf{82}	&	90	\\
\texttt{Spd\_RF2\_300\_353\_4275}	&	77	&	76	&	  \texttt{Spd\_RF2\_450\_614\_4979}	&	\textbf{80}	&	88	\\
\texttt{Spd\_RF2\_300\_353\_4283}	&	\textbf{75}	&	76	&	  \texttt{Spd\_RF2\_450\_614\_4987}	&	\textbf{79}	&	85	\\
\texttt{Spd\_RF2\_300\_380\_4291}	&	\textbf{65}	&	69	&	  \texttt{Spd\_RF2\_450\_614\_4995}	&	\textbf{79}	&	86	\\
\texttt{Spd\_RF2\_300\_380\_4299}	&	66	&	66	&	  \texttt{Spd\_RF2\_450\_614\_5003}	&	\textbf{83}	&	87	\\
\texttt{Spd\_RF2\_300\_380\_4307}	&	\textbf{62}	&	64	&	  \texttt{Spd\_RF2\_500\_534\_5011}	&	147	&	145	\\
\texttt{Spd\_RF2\_300\_380\_4315}	&	\textbf{58}	&	63	&	  \texttt{Spd\_RF2\_500\_534\_5019}	&	148	&	147	\\
\texttt{Spd\_RF2\_300\_380\_4323}	&	\textbf{62}	&	66	&	  \texttt{Spd\_RF2\_500\_534\_5027}	&	150	&	146	\\
\texttt{Spd\_RF2\_300\_407\_4331}	&	\textbf{56}	&	60	&	  \texttt{Spd\_RF2\_500\_534\_5035}	&	150	&	148	\\
\texttt{Spd\_RF2\_300\_407\_4339}	&	60	&	58	&	  \texttt{Spd\_RF2\_500\_534\_5043}	&	147	&	145	\\
\texttt{Spd\_RF2\_300\_407\_4347}	&	53	&	63	&	  \texttt{Spd\_RF2\_500\_568\_5051}	&	129	&	128	\\
\texttt{Spd\_RF2\_300\_407\_4355}	&	\textbf{53}	&	56	&	  \texttt{Spd\_RF2\_500\_568\_5059}	&	\textbf{129}	&	132	\\
\texttt{Spd\_RF2\_300\_407\_4363}	&	\textbf{56}	&	59	&	  \texttt{Spd\_RF2\_500\_568\_5067}	&	\textbf{131}	&	132	\\
\texttt{Spd\_RF2\_300\_434\_4371}	&	\textbf{46}	&	50	&	  \texttt{Spd\_RF2\_500\_568\_5075}	&	133	&	131	\\
\texttt{Spd\_RF2\_300\_434\_4379}	&	\textbf{45}	&	47	&	  \texttt{Spd\_RF2\_500\_568\_5083}	&	\textbf{130}	&	132	\\
\texttt{Spd\_RF2\_300\_434\_4387}	&	\textbf{46}	&	47	&	  \texttt{Spd\_RF2\_500\_603\_5091}	&	\textbf{118}	&	125	\\
\texttt{Spd\_RF2\_300\_434\_4395}	&	\textbf{48}	&	53	&	  \texttt{Spd\_RF2\_500\_603\_5099}	&	116	&	115	\\
\texttt{Spd\_RF2\_300\_434\_4403}	&	\textbf{49}	&	56	&	  \texttt{Spd\_RF2\_500\_603\_5107}	&	122	&	121	\\
\texttt{Spd\_RF2\_350\_378\_4411}	&	100	&	99	&	  \texttt{Spd\_RF2\_500\_603\_5115}	&	\textbf{116}	&	123	\\
\texttt{Spd\_RF2\_350\_378\_4419}	&	99	&	96	&	  \texttt{Spd\_RF2\_500\_603\_5123}	&	\textbf{116}	&	117	\\
\texttt{Spd\_RF2\_350\_378\_4427}	&	102	&	100	&	  \texttt{Spd\_RF2\_500\_637\_5131}	&	\textbf{106}	&	112	\\
\texttt{Spd\_RF2\_350\_378\_4435}	&	99	&	97	&	  \texttt{Spd\_RF2\_500\_637\_5139}	&	\textbf{105}	&	108	\\
\texttt{Spd\_RF2\_350\_378\_4443}	&	100	&	98	&	  \texttt{Spd\_RF2\_500\_637\_5147}	&	\textbf{102}	&	107	\\
\texttt{Spd\_RF2\_350\_406\_4451}	&	\textbf{89}	&	91	&	  \texttt{Spd\_RF2\_500\_637\_5155}	&	\textbf{101}	&	106	\\
\texttt{Spd\_RF2\_350\_406\_4459}	&	89	&	87	&	  \texttt{Spd\_RF2\_500\_637\_5163}	&	\textbf{97}	&	98	\\
\texttt{Spd\_RF2\_350\_406\_4467}	&	\textbf{88}	&	91	&	  \texttt{Spd\_RF2\_500\_672\_5171}	&	\textbf{93}	&	105	\\
\texttt{Spd\_RF2\_350\_406\_4475}	&	87	&	85	&	  \texttt{Spd\_RF2\_500\_672\_5179}	&	\textbf{90}	&	98	\\
\texttt{Spd\_RF2\_350\_406\_4483}	&	\textbf{87}	&	90	&	  \texttt{Spd\_RF2\_500\_672\_5187}	&	\textbf{89}	&	92	\\
\texttt{Spd\_RF2\_350\_435\_4491}	&	\textbf{75}	&	77	&	  \texttt{Spd\_RF2\_500\_672\_5195}	&	\textbf{91}	&	103	\\
\texttt{Spd\_RF2\_350\_435\_4499}	&	75	&	74	&	  \texttt{Spd\_RF2\_500\_672\_5203}   &	  \textbf{89}  &   97  \\
\texttt{Spd\_RF2\_350\_435\_4507}	&	\textbf{71}	&	76	&	      &        & \\	
\bottomrule
\end{longtable}
\end{scriptsize}


\subsection{Instances from Silva et al.~\cite{Siletal14}}
\label{sec:xp-silva}

The computational results of Silva et al.~\cite{Siletal14} include six 
classes of benchmark instances. 
We remark that only three of those sets, which we describe next, were available for our experiments.

\begin{description}
\item[Set III:] includes four instances adapted from the TSPLIB, ranging from 1,000 to 4,000 vertices and from 1,998 to 7,997 edges.
\item[Set V:] includes five instances adapted from the OR-Library, all with 1,000 vertices and 5,000 edges.
\item[Set VI:] includes twelve instances proposed by Leighton~\cite{Lei79}, all with 450 vertices and an edge count varying from 5,714 to 17,425.
\end{description}

Although these are relatively large instances, all of them have optimal solutions 
with no branch vertices that could be obtained with our enhanced algorithm.
Table \ref{tab:silvaInstances} indicates the execution 
time (in seconds) to solve those instances with the enhanced algorithm,
and also compare the heuristic solution values achieved with the best results presented by Silva et 
al.~\cite{Siletal14}. Note that the latter is a randomized algorithm, and the 
authors report the minimum, maximum and average solution values over 100 
executions. Nevertheless, we present in the table only their results with minimum number of branch vertices.

The path expanding and multi-path expanding heuristics were able to construct 
much superior solutions.
In fact, they are able to find provably optimal (\emph{i.e.} branch-free) 
solutions in many cases.
Note that, for all instances not optimally solved by~\cite{Siletal14}, both our solutions improve their results.

Finally, since all the instances in these benchmark sets have optimal solutions with no branch vertices, the obligatory branches lower bound has no effect. As for the cut edges decomposition, only one of the instances (\texttt{VI/le450\_15b}) has two bridges.

\begin{scriptsize}
\begin{longtable}[HT]{@{}llllcccc@{}}
\caption{Results regarding the available instances used by Silva et al.~\cite{Siletal14}. All these instances admit a branch-free solution, which could be obtained using the enhanced algorithm in the time indicated in the fourth column. \label{tab:silvaInstances}}\\
\toprule									
                        	 &	\multicolumn{3}{c}{Instance}	 &	Exact solution	 & Path         & Multi-Path    &	Best result	  \\
                        	 &	id & $n$ & $m$                   &	time (s)	      &	Expanding   & Expanding     &	by \cite{Siletal14}	  \\
\midrule
\addlinespace[11pt]
\multirow{4}{*}{Set III}	 &	alb1000	 &	1000	 &	1998	 &	30.13	 &	15	 &	16	 &	54	  \\
                        	 &	alb2000	 &	2000	 &	3996	 &	180.18	 &	26	 &	28	 &	121	  \\
                        	 &	alb3000a	 &	3000	 &	5999	 &	74.70 &	50	 &	43	 &	191	  \\
                        	 &	alb4000	 &	4000	 &	7997	 &	1778.87	 &	65	 &	58	 &	247	  \\
\addlinespace[11pt]															
\multirow{5}{*}{Set V}	 &	steind11	 &	1000	 &	5000	 &	0.32	 &	0	 &	0	 &	33	  \\
                    	 &	steind12	 &	1000	 &	5000	 &	30.13	 &	1	 &	1	 &	26	  \\
                    	 &	steind13	 &	1000	 &	5000	 &	0.34	 &	1	 &	0	 &	28	  \\
                    	 &	steind14	 &	1000	 &	5000	 &	41.09	 &	1	 &	1	 &	28	  \\
                    	 &	steind15	 &	1000	 &	5000	 &	57.71	 &	2	 &	1	 &	27	  \\
\addlinespace[11pt]															
\multirow{12}{*}{Set VI}	 &	le450\_5a	 &	450	 &	5714	 &	2.25	 &	0	 &	0	 &	1	  \\
                        	 &	le450\_5b	 &	450	 &	5734	 &	2.15	 &	0	 &	0	 &	1	  \\
                        	 &	le450\_5c	 &	450	 &	9803	 &	0.23	 &	0	 &	0	 &	0	  \\
                        	 &	le450\_5d	 &	450	 &	9757	 &	4.64	 &	0	 &	0	 &	0	  \\
                        	 &	le450\_15a	 &	450	 &	8186	 &	0.27	 &	0	 &	0	 &	4	  \\
                        	 &	le450\_15b	 &	450	 &	8169	 &	2.23	 &	1	 &	1	 &	3	  \\
                        	 &	le450\_15c	 &	450	 &	16680	 &	0.47	 &	0	 &	0	 &	0	  \\
                        	 &	le450\_15d	 &	450	 &	16750	 &	0.82	 &	0	 &	0	 &	0	  \\
                        	 &	le450\_25a	 &	450	 &	8160	 &	0.94	 &	0	 &	0	 &	8	  \\
                        	 &	le450\_25b	 &	450	 &	8263	 &	0.88	 &	0	 &	0	 &	4	  \\
                        	 &	le450\_25c	 &	450	 &	17343	 &	1.53	 &	0	 &	0	 &	0	  \\
                        	 &	le450\_25d	 &	450	 &	17425	 &	0.53	 &	0	 &	0	 &	0     \\	
\bottomrule									
\end{longtable}
\end{scriptsize}

\section{Final remarks}
\label{sec:finalremarks}

This paper introduces an effective decomposition method and two constructive heuristics for the minimum branch vertices problem.
Since most benchmark instances for the problem (535 out of 546) could be solved to optimality by a branch and cut algorithm, 
we stress the relevance of the algorithms introduced here as preprocessing methods: 
a phase between formulation and solution for improving the algorithmic solvability of the problem \cite{NemhauserWolsey1988}.
We highlight the computational efficiency of the algorithms we present, whose implementations run in less than a second for all available instances.

We compared a standard branch and cut algorithm with its application on the remaining subproblems after running our decomposition and heuristic methods. Not only does the \emph{enhanced version} provide a better duality gap in $96\%$ of the instances with no optimality certificate, as it makes the algorithm consistently faster in all cases.

The decomposition method is fast and effectively reduces problem instances:
we present average results ranging from 15.9\% to 24.8\% of removed vertices, and from 38.6\% to 61.3\% of removed edges.

Finally, the heuristics provided better MIP starts in most cases: it is better than the ones presented by Silva et al.~\cite{Siletal14} in all of the 21 available instances, among those used in their experiments. Our heuristics also provide better primal bounds for 103 out of 175 instances, for which Carrabs et al.~\cite{Caretal13} describe extended results.
We remark that our proposed heuristics could be used to rapidly provide very good quality solutions for more advanced local search procedures, such as the one recently proposed by Mar\'in~\cite{Mar15}.

\bibliographystyle{plain}

\end{document}